\documentclass[a4paper,UKenglish]{lipics-v2018}
\usepackage{microtype}%if unwanted, comment out or use option "draft"

\title{Optimal Rank and Select Queries on Dictionary-Compressed Text}
\titlerunning{Rank and Select on Dictionary-Compressed Text}

\author{Nicola Prezza}{Department of Computer Science, University of Pisa, Italy}{nicola.prezza@di.unipi.it}{https://orcid.org/0000-0003-3553-4953
}{}%mandatory, please use full name; only 1 author per \author macro; first two parameters are mandatory, other parameters can be empty.

\authorrunning{N. Prezza}%mandatory. First: Use abbreviated first/middle names. Second (only in severe cases): Use first author plus 'et al.'

\Copyright{Nicola Prezza}%mandatory, please use full first names. LIPIcs license is "CC-BY";  http://creativecommons.org/licenses/by/3.0/

\subjclass{}

%mandatory

\category{}%optional, e.g. invited paper

\relatedversion{}%optional, e.g. full version hosted on arXiv, HAL, or other respository/website

\supplement{}%optional, e.g. related research data, source code, ... hosted on a repository like zenodo, figshare, GitHub, ...

\funding{}%optional, to capture a funding statement, which applies to all authors. Please enter author specific funding statements as fifth argument of the \author macro.

\EventEditors{John Q. Open and Joan R. Access}
\EventNoEds{2}
\EventLongTitle{42nd Conference on Very Important Topics (CVIT 2016)}
\EventShortTitle{CVIT 2016}
\EventAcronym{CVIT}
\EventYear{2016}
\EventDate{December 24--27, 2016}
\EventLocation{Little Whinging, United Kingdom}
\EventLogo{}
\SeriesVolume{42}
\ArticleNo{23}
\nolinenumbers %uncomment to disable line numbering
\begin{document}

\maketitle

%\begin{abstract}
%Let $\gamma$ be the size of a string attractor for a string $S$ of length $n$ over an alphabet of size $\sigma$. By known reductions, one can build a string attractor such that $\gamma$ is upper-bounded by the output size of any known dictionary compressor (e.g. grammars, LZ77, or the run-length Burrows-Wheeler transform). It is known that, within $O\left (\gamma \log^\epsilon n\log(n/\gamma) /\log\log n\right)$ words of space, random access on $S$ can be performed in optimal $O(\log(n/\gamma)/\log\log n)$ time. In this paper we extend this functionality to \emph{rank} and \emph{select} queries. Our solutions use a factor $\sigma$ more space than the solution for \emph{access}, but preserve the same query times. We provide lower bounds showing that these query times are optimal when $\sigma \in O(\rm{polylog}\ n)$. Differently from previous solutions on dictionary-compressed text, we reduce \emph{rank} and \emph{select} to the \emph{partial sum} problem, for which we give a provably optimal-time solution within $O\left (\gamma \log^\epsilon n\log(n/\gamma) /\log\log n\right)$ words of space. 
%We give our results in the form of a space-time trade-off that is more general than the best previous one on grammars and that improves existing bounds on LZ77-compressed text by a $\log\log n$ time-factor in \emph{select} queries.
%As a result, we also obtain the first solution achieving constant time for \emph{rank} and \emph{select} queries within $O(\sigma n^\epsilon \gamma^{1-\epsilon})$ space. 
%\end{abstract}

\begin{abstract}
	We study the problem of supporting queries on a string $S$ of length $n$ within a space bounded by the size $\gamma$ of a string attractor for $S$.
	Recent works showed that random access on $S$ can be supported in optimal $O(\log(n/\gamma)/\log\log n)$ time within $O\left (\gamma\ \rm{polylog}\ n \right)$ space.
	%In the original paper introducing these combinatorial objects it was shown that $\gamma$ can be upper-bounded by the size of most known dictionary compressors and that random access on $S$ can be supported in optimal $O(\log(n/\gamma)/\log\log n)$ time within $O\left (\gamma\ \rm{polylog}\ n \right)$ space. 
	In this paper, we extend this result to 
	\emph{rank} and \emph{select} queries and provide lower bounds matching our upper bounds on alphabets of polylogarithmic size. 
	Our solutions are given in the form of a space-time trade-off that is more general than the one previously known for grammars and that improves existing bounds on LZ77-compressed text by a $\log\log n$ time-factor in \emph{select} queries.
	We also provide matching lower and upper bounds for \emph{partial sum} and \emph{predecessor} queries within attractor-bounded space, and extend our lower bounds to encompass navigation of dictionary-compressed tree representations. 
	%In these cases, our solutions are optimal regardless the alphabet's size.  
	%on alphabets of polylogarithmic size. 
	%To conclude, we extend our lower bounds to queries on dictionary-compressed strings of balanced parentheses, showing that 
	%and provide attractor-based data structures matching the lower bounds. This yields the first optimal-time solutions for grammar-compressed trees, and the first solutions for more powerful compression schemes such as LZ77 and Macro Schemes.  
\end{abstract}

\section{Related Work}

Access, rank, and select queries stand at the core of many tasks on compact and compressed data structures, including compressed indexes, graphs, trees, sets of points, etc. Given a string $S[1..n]$  over an integer alphabet $\Sigma = [0,\sigma-1]$, these queries are defined as:

\begin{itemize}
	\item $S.access(i)$, $i=1,\dots, n$: return the $i$-th symbol of $S$.
	\item $S.rank_c(i)$, $i=1,\dots, n$: return the number of occurrences of symbol $c$ in $S[1..i]$.
	\item $S.select_c(i)$, $i=1,\dots, S.rank_c(n)$: return the position of the $i$-th occurrence of $c$ in $S$.
\end{itemize}

While the problem is essentially solved within entropy-compressed bounds~\cite{BN14}, the increasingly growing production of repetitive datasets in fields such as biology and physics (as well as in many web software repositories) is raising the problem of extending such functionalities (as well as more complex queries such as indexing) to dictionary-compressed representations. 
This problem has lately received some attention in the literature, and solutions are known for some (but not all) compression schemes. 
Belazzougui et al. in~\cite{belazzougui2015queries} provided near-optimal bounds on LZ77-compressed text: letting $z$ be the number of LZ77 phrases, they show how to perform access in $O\left (z \log^\epsilon n\log(n/z)/\log\log n\right)$ space and $O(\log(n/z)/\log\log n)$ time. \emph{rank} and \emph{select} require a $\sigma$-factor more space and are supported in $O(\log(n/z)/\log\log n)$ and $O(\log(n/z))$ time, respectively. 
Ord{\'o}{\~n}ez et al.~\cite{ordonez2017grammar} and Belazzougui at al.~\cite{belazzougui2015access} solved the problem on grammar-compressed strings in $O(\sigma g)$ space and $O(\log n)$ time for rank and select queries. The latter paper also provides a trade-off using $O(\tau\sigma g \log_\tau(n/g))$ words and supporting queries in $O(\log_\tau(n/g))$ time, for $2\leq \tau \leq \log^\epsilon n$ and any constant $\epsilon>0$. For $\tau=\log^\epsilon n$, this solution yields $O(\sigma g \log^\epsilon n\log(n/g)/\log\log n)$ space and $O(\log(n/g)/\log\log n)$ query time. No solutions are known for other dictionary compression schemes such as Macro Schemes~\cite{storer1982data}, Collage Systems~\cite{KidaMSTSA03}, or the run-length Burrows-Wheeler transform~\cite{burrows1994block}.

In this paper, we use the newborn theory of string attractors~\cite{kempa2018string,kempa2018roots} to provide a universal solution working simultaneously on all known dictionary compressors. We moreover provide the first lower bounds for these queries, which match our upper-bounds on polylogarithmic alphabets. Our solutions are the first for the run-length Burrows-Wheeler transform, Macro Schemes, Run-length SLPs, and Collage Systems. We obtain the full-spectrum trade-off for grammars, generalizing the result in~\cite{belazzougui2015access} for all $2\leq \tau \leq n$, and improve the existing bounds on LZ77-compressed text~\cite{belazzougui2015queries} by a $\log\log n$ time-factor in \emph{select} queries. 

Importantly, we note that the reason why Belazzougui et al. cannot reach the optimal $O(\log(n/z)/\log\log n)$ time for \emph{select} queries on LZ77 in~\cite{belazzougui2015queries} and cannot get the full spectrum $2\leq \tau \leq n$ on grammars in~\cite{belazzougui2015access} is the same: they need to perform a predecessor query at each level of their structure. In the first case, they use z-fast tries~\cite{belazzougui2009monotone}, which introduce a $O(\log\log n)$ multiplicative factor in query times. In the second case, they use fusion trees~\cite{fredman1990blasting}, which support predecessor queries in constant time only for $\tau = O(\log^\epsilon n)$. We solve this problem, and obtain optimal query times also for \emph{select} queries, by reducing both \emph{rank} and \emph{select} to partial sum queries:

\begin{itemize}
	\item $S.psum(i) = \sum_{j=1}^i S[j],\ i=1,\dots, n$.
\end{itemize}

We show how to support partial sums in optimal time within attractor-bounded space (optimality follows from a straightforward reduction from \emph{access} queries) by generalizing the strategy used by Belazzougui et al.~\cite{belazzougui2015queries} to solve \emph{rank} queries. Importantly, our reductions to \emph{partial sum} preserve the attractor size. This solution allows us to avoid performing expensive predecessor queries at each level of our structure, thus obtaining constant time per level on the whole range $2\leq \tau \leq n$. 

On our way, we extend our lower bounds to operations on dictionary-compressed sequences of balanced parentheses, typically used to support navigation on (compressed) trees. Our lower bounds show that existing solutions~\cite[Lem. 8.2, 8.3]{bille2015random} on trees represented using grammar-compressed sequences of balanced parentheses are far from the optimum by a $O(\log\log n)$-factor.

We also note that our lower- and upper- bounds easily extend to the well-studied \emph{predecessor} problem, which asks to find the largest element $x$ not larger than a given $y$ in an opportunely-encoded set $\{x_1, \dots, x_m\} \subseteq [1,n]$. 
A classic result from Beame and Fich~\cite[Cor 3.10]{beame2002optimal} states that, using words of size $\log^{O(1)} n$, no static data structure of size $m^{O(1)}$ can answer predecessor queries in time $o(\sqrt{\log m/\log\log m})$. First, note that a dictionary compressed representation of the sequence $x_1, x_2-x_1, \dots, x_m-x_{m-1}$ always takes at most $O(m)$ words of space, and could take much less if the sequence is repetitive. Indeed, we show that Beame and Fich's lower bound can be improved to $\Omega(\log n/\log\log n)$ when the sequence is dictionary-compressed, and provide a data structure matching this lower bound.

\subsection{String Attractors}

Let $S$ be a string of length $n$. Informally, a string attractor~\cite{kempa2018roots} for $S$ is a set $\Gamma \subseteq [1..n]$ with the following property: any substring of $S$ has at least one occurrence in $S$ crossing at least one position in $\Gamma$. The following definition formalizes this concept. 

\begin{definition}[String attractor~\cite{kempa2018roots}]\label{def: string attractor}
	A \emph{string attractor} of a string $S\in\Sigma^n$ is a set of
	positions $\Gamma \subseteq [1..n]$ such that every substring
	$S[i..j]$ has at least one occurrence
	$S[i'..j'] = S[i..j]$ with $j'' \in [i'..j']$ for some
	$j''\in\Gamma$.
\end{definition}

String attractors were originally introduced as a unifying framework for known dictionary compressors: Straight-Line
programs~\cite{KY00} (context-free grammars generating the string), Collage Systems~\cite{KidaMSTSA03},
Macro schemes~\cite{storer1982data} (a set of substring equations
having the string as unique solution; this includes Lempel-Ziv 77~\cite{lempel1976complexity}), the run-length Burrows-Wheeler
transform~\cite{burrows1994block} (a string permutation whose number
of equal-letter runs decreases as the string's repetitiveness
increases), and the compact directed acyclic word
graph~\cite{blumer1987complete,crochemore1997direct} (the minimization
of the suffix tree). As shown in~\cite{kempa2018roots}, any of the above compressed representations induces a string attractor of the same asymptotic size, which for most compressors is also a polylogarithmic approximation to the smallest attractor (NP-hard to find~\cite{kempa2018string,kempa2018roots}).
The other way round also holds for a subset of the above compressors: given a string attractor of size $\gamma$, one can build a compressed representation of size $O(\gamma\log(n/\gamma))$. 
These reductions imply that we can design universal compressed data structures (i.e. working on top of any of the above compressors). In particular, it can be shown that optimal-time random access can be supported within $O(\gamma\ \mathrm{polylog}\ n)$ space~\cite{kempa2018roots}. Similarly, fast text indexing can be achieved within the same space~\cite{navarro2018faster,navarro2018universal}.

%In the following sections we extend the range of queries that can be performed optimally in attractor-bounded space to  \emph{rank}, \emph{select}, \emph{partial sum}, and \emph{predecessor}. 

%Computing a minimum $k$-attractor is NP-complete for general alphabets and $k\geq 3$, and APX-complete for constant $k\geq 3$~\cite{kempa2018roots}. If the alphabet and $k$ are small, however, the problem admits polynomial-time solutions~\cite{kempa2018string}.

%Verification and approximation, on the other hand, are much easier problems: we can decide whether $\Gamma$ is a proper $k$-attractor in optimal $O(n)$ time and space~\cite{kempa2018string} and, with a simple reduction to set-cover, we can compute a $O(\log k)$ approximation of the smallest $k$ -attractor in polynomial time~\cite{kempa2018roots}. 

%\section{Queries on Dictionary-Compressed Strings}

\section{Lower Bounds on Dictionary-Compressed Strings}\label{sec:lower string}

We start by providing lower bounds for \emph{rank, select} and \emph{partial sum} queries on dictionary-compressed text. 
We also consider \emph{predecessor} queries on sets represented by dictionary-compressed binary strings. 
Our lower bounds are shown using grammar compression, and therefore automatically extend to any compression scheme more powerful than SLPs. In the following, we only consider grammars whose right-hand side has size two, and define the grammar's size to be the number of nonterminals. 

Our starting point is the following theorem from Verbin and Yu~\cite{CVY13}, in the variant revisited by Kempa and Prezza~\cite[Thm 5.1]{kempa2018roots}:

\begin{theorem}[Verbin and Yu~\cite{CVY13}]\label{th:low_bound_access}
	Let $g$ be the size of any Straight-Line Program for a string $S$ of length $n$ over a binary alphabet. 
	Any static data structure taking $O(g\ {\rm polylog}\ n)$ space cannot answer random access queries on $S$ in less than $O(\log n / \log \log n)$ time.
\end{theorem}

The idea is to show a reduction from \emph{access} queries to \emph{partial sum}, \emph{predecessor}, and \emph{rank} and from \emph{rank} to \emph{select}, while asymptotically preserving the grammar size. For \emph{rank} and \emph{partial sum} the reduction is straightforward. Let $S$ be a binary string. Then, $S.access(i) = S.rank_1(i)-S.rank_1(i-1) = S.psum(i)-S.psum(i-1)$, where $S.rank(0)=S.psum(0)=0$ for convenience. It follows that also $rank_1$ and \emph{partial sum} queries cannot break the lower bound of Theorem \ref{th:low_bound_access} (note that, since this is a lower bound, we can remove the restriction on the alphabet size).

To extend the lower bound to \emph{select} queries, we build a string $\delta(S)$ such that (i) $\delta(S)$ has a SLP of size at most $g+1$, and (ii) \emph{rank} queries on $S$ can be simulated with a constant number of \emph{select} queries on $\delta(S)$. Then, the result follows from the hardness of \emph{rank}. 

\begin{definition}
	Let $S$ be a binary string of length $n$. With $\delta:\{0,1\}\rightarrow \{0,1\}^*$ we denote the function defined as $\delta(0) = 1$ and $\delta(1)=01$. With $\delta(S)$ we denote the string $\delta(S[1])\delta(S[2])\dots \delta(S[n])$.
\end{definition}

\begin{lemma}\label{lem:deltaS_attractor}
	If $S$ has a Straight Line Program of size $g$, then $\delta(S)$ has a Straight Line Program of size at most $g+1$.
\end{lemma}
\begin{proof}
	It is sufficient to modify the Straight Line Program $G$ for $S$ as follows. First, we create a new nonterminal $X$ expanding to $01$ (this is not necessary if such a nonterminal already exists). Then, in the rules of $G$, we replace each terminal $1$ with $X$ and each terminal $0$ with $1$. It is easy to see that the resulting SLP --- of size at most $g+1$ --- generates $\delta(S)$.
\end{proof}

\begin{lemma}\label{lem:deltaS_select}
	$S.rank_1(i) = \delta(S).select_1(i)-i$.
\end{lemma}
\begin{proof}
	First, note that each character of $S$ generates exactly one bit set in $\delta(S)$.
	Then, $\delta(S).select_1(i)$ is the position of the last bit of the encoding of $S[i]$ in $\delta(S)$.
	
	Moreover, each bit equal to $1$ in $S$ generates exactly one $0$-bit in $\delta(S)$, while $0$-bits in $S$ do not generate $0$-bits in $\delta(S)$. 
	Then, the number of $0$'s before position $t$ in $\delta(S)$ --- which is $t-i = \delta(S).select_1(i)-i$ --- corresponds to $S.rank_1(i)$, i.e. our claim.
\end{proof}

The above lemmas imply that also $select_1$ queries cannot break the lower bound within grammar-compressed space.
Note that our lower bounds can trivially be extended to $rank_0$ and $select_0$ by simply flipping all bits (this operation does not increase the grammar size as it is sufficient to flip the two grammar's terminals).

\begin{theorem}\label{thm:low_bounds_g}
	Let $S$ be a string of length $n$, and let $g$ be the size of a Straight-Line Program for $S$.
	Then, $\Omega(\log n/\log\log n)$ time is needed to perform partial sum, rank, and select queries on $S$ within $O(g\ \mathrm{polylog}\ n)$ space.
\end{theorem}

We now move to the well-studied \emph{predecessor} problem. Let $U\subseteq [1,n]$ be a set of integers of cardinality $m$. 
%If the distances between consecutive elements of $U$ form a repetitive sequence, then one could hope to compress $U$ below the information-theoretic bound of $\Theta(m\log(n/m))$ bits by means of dictionary compressors. 
Beame and Fich~\cite[Cor 3.10]{beame2002optimal} proved that, using words of size $\log^{O(1)} n$, no static data structure of size $m^{O(1)}$ can answer predecessor queries in time $o(\sqrt{\log m/\log\log m})$. 
%Now, consider representing $U$ with a binary string $S$ of length $n$ such that $S[i] = 1$ iff $i\in U$. As noted above, regularities in the distances between elements of $U$ result in a more compressible string $S$.
Note that the size $g$ of a straight-line program expanding to the distances between elements of $U$ could be much smaller than $m$; one of the consequences of this increased compression power is that we can improve Beame and Fich's lower bound within space bounded by $g$:

\begin{theorem}\label{th:low bound pred}
	Let $U\subseteq [1,n]$ be a set of size $m$, and let $S\in \{0,1\}^n$ be the bit-string representing $U$: $S[i] = 1$ iff $i\in U$. Let moreover $g$ be the size of a Straight-Line Program for $S$.
	Then, $\Omega(\log n/\log\log n)$ time is needed to perform predecessor queries on $U$ within $O(g\ \mathrm{polylog}\ n)$ space.
\end{theorem}
\begin{proof}
	The proof is a straightforward reduction from \emph{access} queries. Let $S$ be a binary string, and $U$ be the set defined as $i\in U$ iff $S[i] = 1$. Then, for $i<n$, $S[i] = 1$ if and only if the predecessor of $i+1$ in $U$ is $i$. The lower bound follows from Theorem \ref{th:low_bound_access}.
\end{proof}

As noted in~\cite{kempa2018roots}, all the above lower bounds immediately extend to any compression scheme more powerful than grammars (including string attractors).

\begin{corollary}\label{cor:low_bound_compressors}
	Let $S$ be a string of length $n$, and let $\alpha$ be any of these measures of repetitivity of $S$: the size $\gamma$ of a string attractor, the size $g$ of a SLP, the size $g_{rl}$ of a RLSLP, the size $c$ of a collage system, the size $z$ of the LZ77 parse, the size $b$ of a macro scheme.
	Then, $\Omega(\log n/\log\log n)$ time is needed to perform partial sum, rank, and select queries on $S$ within $O(\alpha\ \mathrm{polylog}\ n)$ space. The same lower bound holds for \emph{predecessor} queries when $S$ is binary and represents a set of integers.
\end{corollary}

\section{Lower bounds on Dictionary-Compressed Trees}\label{sec:lower tree}

A space-efficient way to represent trees, that also supports fast navigational operations, is to encode their topology with a (binary) string. There are three principal ways to do this: LOUDS~\cite{jacobson1989space}, DFUDS~\cite{benoit2005representing}, and BP~\cite{jacobson1989space}. 
LOUDS requires just \emph{rank/select} support on a binary string, for which we already provided matching lower- and upper-bounds.  
We now extend our lower bounds to operations on balanced parentheses (DFUDS/BP), which will show the hardness of navigating a dictionary-compressed tree representation. 
DFUDS and BP require additional primitives to be supported on the underlying string of balanced parentheses:

\begin{enumerate}
	\item[(1)] \texttt{excess($i$)}: the difference between open and closed parentheses before position $i$ (included).
	\item[(2)] \texttt{findopen($i$)}: the position $j$ of the closed parenthesis matching the open parenthesis in position $i$. 
	\item[(3)] \texttt{findclose($i$)}: the inverse of \texttt{findopen}.
	\item[(4)] \texttt{fwd\_search($i,\delta$)}: the first position $j>i$ such that  \texttt{excess($j$)} = \texttt{excess($i$)$+\delta$}.
	\item[(5)] \texttt{bwd\_search($i,\delta$)}: as  \texttt{fwd\_search}, but looking backwards.
	\item[(6)] \texttt{rmq($i,j$)}/\texttt{RMQ($i,j$)}: minimum/maximum in \texttt{excess($i..j$)}.
	\item[(7)] \texttt{rmqi($i,j$)}/\texttt{RMQi($i,j$)}: leftmost position of a minimum/maximum in \texttt{excess($i..j$)}.
\end{enumerate}

For more details, Navarro~\cite{navarro2016compact} gives a complete and low-level description of the above operations and of how tree navigation queries can be reduced to them.
%In this section we study the complexity of supporting these operations on a dictionary-compressed representation of a balanced parenthesis encoding (DFUDS or BP) of the tree topology.
Alternative ways to compress trees include \emph{Tree Straight-Line Programs} (TSLPs), which are not covered here\footnote{It is worth to note, however, that the compression ratio of these representations is not as good as that of an SLP for the DFUDS string of small-degree trees~\cite[Thms 2, 4]{ganardi2018tree}.}.

%We first show that, unsurprisingly, the lower bound  $O(\log n/\log\log n)$  holds also for individual operations (1-9) in grammar-compressed space. We extend this to some high-level navigational tree operations ($i$-th child on BP/DFUDS and level ancestor queries on BP), showing that the lower bound holds no matter the implementation used for these queries (i.e. even if operations other than (1-9) are used to implement them). 

%We will first exhibit trees on which LZ77, RLBWT, and Macro Schemes provide better compression ratios than SLPs: this implies that supporting navigational operations on those (more powerful) compressors is of interest. 
%Clearly, our lower bounds extend to those compressors (since they are more powerful than SLPs). 
%We will moreover provide data structures matching our lower bounds for operations (1-9) on an attractor-compressed sequence of balanced parentheses. 
%The use of string attractors automatically extends the validity of our solution to all known dictionary compression schemes. As a result, we obtain matching lower- and upper- bounds for navigational queries on dictionary-compressed trees. 

%\subsection{Excess}

We start by showing a reduction from $rank_1$ to \emph{excess}. Given any binary string $S$, we show how to build a string $\Delta(S)$ of balanced parentheses such that (i) the grammar-compressed representation of  $\Delta(S)$ is not much larger than that of $S$ and (ii) $S.rank_1$ can be solved with a constant number of $\Delta(S).excess$ queries. Note that, in the above definition, we add an extra pair of enclosing parentheses in order to make the sequence a tree (otherwise, the transformed string could represent a forest).

\begin{definition}\label{def:excess hard}
	Let $\delta(0) = $ \texttt{()} and $\delta(1) = $ \texttt{((}. When $S$ is a binary string of length $n$, we define $\Delta(S) = $ \texttt{(} $\cdot\  \delta(S[1])\ \cdots\ \delta(S[n])\ \cdot$ \texttt{)}$^{k+1}$, where $k=2\cdot S.rank_1(n)$.
\end{definition}

\begin{example}
	If $S = 00101$, then $\Delta(S) =$ \texttt{(()()((()(()))))}.
\end{example}

Note that $\Delta(S)$ is always balanced: first, we introduce an open parenthesis, then terms $\delta(0)$ are balanced and terms $\delta(1)$ introduce two unbalanced open parentheses each. Those $k=2\cdot S.rank_1(n)$ parentheses, plus the first open parenthesis are balanced in the final suffix \texttt{)}$^{k+1}$ of $\Delta(S)$. 
%Moreover, there always exist two trees $T_1$ and $T_2$ such that $DFUDS(T_1) = BP(T_2) = \Delta(S)$ \textcolor{red}{TODO proof}.

\begin{lemma}\label{lem:excess lower SLP}
	If $S\in\{0,1\}^n$ has a SLP of size $g$, then $\Delta(S)$ has a SLP of size $O(g+\log n)$.
\end{lemma}
\begin{proof}
	Let $G$ be a SLP for $S$.	
	We replace the terminal '0' with a nonterminal expanding to \texttt{()}, and the terminal '1' with a nonterminal expanding to \texttt{((}. We finally create at most $O(\log k) = O(\log n)$ new nonterminals to generate the final suffix \texttt{)}$^{k+1}$, and add two more rules to concatenate the resulting SLPs to the additional open parenthesis prefixing $\Delta(S)$.
\end{proof}

\begin{lemma}\label{lem:excess lower reduction}
	$S.rank_1(i) = (\Delta(S).excess(2i+1)-1)/2$.
\end{lemma}
\begin{proof}
Assume 	$S.rank_1(i)=t$. Then, $S.rank_0(i)=i-t$. Since each '0' in $S$ generates one open and one close parenthesis in $\Delta(S)$, and each '1' in $S$ generates two open parentheses in $\Delta(S)$ and taking into account the extra open parenthesis at the beginning of $\Delta(S)$, the number of open parentheses before $\Delta(S)[2i+1]$ is $1+ 2\cdot t + 1\cdot(i-t)$. Similarly, the number of close parentheses before $\Delta(S)[2i+1]$ is $i-t$. Then, by definition \emph{excess} is precisely the difference between these two values:  $\Delta(S).excess(2i+1) = 1+ 2\cdot t + 1\cdot(i-t) - (i-t) = 2\cdot t +1 = 2\cdot S.rank_1(i) + 1$. Our claim follows. 
\end{proof}

Suppose, by contradiction, that there is a structure supporting $o(\log n'/\log\log n')$-time \emph{excess} queries on a length-$n'$ sequence $B$ within $O(g'\rm {polylog}\ n')$ space, where $g'$ is the size of any SLP compressing $B$. Then, given any binary string $S\in\{0,1\}^n$ with SLP of size $g$, we can build $\Delta(S)$ of length $n'=\Theta(n)$, which by Lemma \ref{lem:excess lower SLP} has a SLP of size $g' = O(g + \log n') = O(g + \log n)$. By Lemma \ref{lem:excess lower reduction} we can use our hypothetical \emph{excess} structure to answer $rank_1$ queries on $S$ in $o(\log n'/\log\log n')$ = $o(\log n/\log\log n)$ time and $O(g'\rm {polylog}\ n') = O(g\ \rm {polylog}\ n)$ space, which by Theorem \ref{cor:low_bound_compressors} is a contradiction. This completes our hardness proof for \emph{excess}.

%\subsection{Findopen, Findclose, and Searching}

We now reduce $rank_1$ to \emph{findclose}. Given any binary string $S$, we show how to build a string $\Delta(S)$ of balanced parentheses such that (i) the grammar-compressed representation of  $\Delta(S)$ is not much larger than that of $S$ and (ii) $S.rank_1$ can be solved with a constant number of $\Delta(S).findclose$ queries. The solution for \emph{findopen} is symmetric and is not considered here.

\begin{definition}
	Let $\delta(0) = $ \texttt{)} and $\delta(1) = $ \texttt{())}. When $S$ is a binary string of length $n$, we define $\Delta(S) = $ \texttt{(}$^n\ \cdot\ \delta(S[1])\ \cdots\ \delta(S[n])$.
\end{definition}

Note that $\Delta(S)$ is always balanced: we first open $n$ parentheses, and then each term $\delta(S[i])$ adds an unmatched closed parenthesis.

\begin{example}
	If $S = 00101$, then $\Delta(S) =$ \texttt{((((())()))())}.
\end{example}

The proof of the following lemma is analogous to that of Lemma \ref{lem:excess lower SLP}, and for space reasons it is omitted here. 

\begin{lemma}\label{lem:findclose lower SLP}
	If $S\in\{0,1\}^n$ has a SLP of size $g$, then $\Delta(S)$ has a SLP of size $O(g+\log n)$.
\end{lemma}

We obtain the following reduction:

\begin{lemma}\label{lem:findclose lower reduction}
	$S.rank_1(i) = \big(\Delta(S).findclose(n-i+1) - n - i\big)/2$.
\end{lemma}
\begin{proof}
	The idea behind the proof is the following. To solve rank, we first ``jump'' on the $i$-th last open parentheses $\Delta(S)[n-i+1]$ in the prefix of length $n$ of $\Delta(S)$. Then, the corresponding closed parentheses $\Delta(S).findclose(n-i+1)$ is the last parenthesis of $\delta(S[i])$ (note that each character of $S$ generates exactly one locally-unmatched closed parenthesis in $\Delta(S)$, which is matched in the prefix \texttt{(}$^n$).
	Let $S.rank_1(i)=t$.
	It follows that before (and including) position $\Delta(S).findclose(n-i+1)$ there are: $n$ parentheses (the prefix \texttt{(}$^n$ of $\Delta(S)$), plus $3t$ parentheses (three for each '1' in $S[1..i]$), plus $i-t$ parentheses (one for each '0' in $S[1..i]$). We conclude that  $\Delta(S).findclose(n-i+1) = n + 3t + (i-t) = n + 2t + i$. Our claim follows.
\end{proof}

Note that operation \emph{findopen} is symmetric to \emph{findclose}, so its hardness is immediate. Finally, the lower bounds automatically transfer to \texttt{fwd\_search} and \texttt{bwd\_search} since they can be used to implement \emph{findopen} and \emph{findclose} (see also \cite{navarro2016compact}):
\texttt{findclose($i$)} = \texttt{fwd\_search($i,-1$)} and \texttt{findopen($i$)} = \texttt{bwd\_search($i,0$)$+1$}.

To prove the hardness of range queries, consider the string $\Delta(S)$ of Definition \ref{def:excess hard}. Clearly, the maximum excess in $\delta(0)$ is reached in the first parenthesis. On the other hand, the maximum excess in $\delta(1)$ is reached in the second parenthesis. This shows that \texttt{RMQi} and \texttt{rmqi} can be used to answer \emph{access} queries: the maximum (resp. minimum) excess in the length-$2$ substring $\Delta(S)[k,k+1]$ corresponding to $\delta(S[i])$ is in position $k$ if and only if $S[i] = 0$ (resp. 1). Similarly, we can solve \texttt{RMQi} (resp. \texttt{rmqi}) by issuing two \texttt{RMQ} (resp. \texttt{rmq}) in the unary ranges $\Delta(S)[k]$ and $\Delta(S)[k+1]$. We finally obtain our result, stated in the most general form: 

\begin{theorem}\label{cor:low_bound_trees}
Let $S$ be a balanced parentheses sequence of length $n$, and let $\alpha$ be any of these measures of repetitivity of $S$: the size $\gamma$ of a string attractor, the size $g$ of a SLP, the size $g_{rl}$ of a RLSLP, the size $c$ of a collage system, the size $z$ of the LZ77 parse, the size $b$ of a macro scheme.
Then, $\Omega(\log n/\log\log n)$ time is needed to perform operations (1-7) on $S$ within $O(\alpha\ \mathrm{polylog}\ n)$ space.
\end{theorem}

\section{Upper Bounds}

In this section we provide upper bounds matching our lower bounds on all known dictionary-compressed string representations. 

\subsection{Partial Sums}

We support \emph{partial sum} queries by generalizing the \emph{rank} solution presented by Belazzougui et al.~\cite{belazzougui2015queries} (designed om block trees) as follows. We divide the text in $\gamma$ blocks of length $n/\gamma$. We call this the \emph{level 0} of our structure. 
We keep $\log_\tau(n/\gamma)$ further levels: for each attractor position $i$, at level $j\geq 1$ we store some information about $2\tau$ non-overlapping and equally-spaced substrings of $S$ (we call them blocks) centered on $i$ and whose length exponentially decreases with the level $j$ (i.e. at level $j=0$ the block length is $n/\gamma$, at level $1$ it is $n/(\tau\gamma)$, at level $2$ it is $n/(\tau^2\gamma)$, and so on). For each block, we store some partial sum information and a pointer to one of its occurrences crossing an attractor position (which exist by definition of attractor). Then, a \emph{partial sum} query is answered by navigating the structure from the first to last level.
%; the main trick consists in ensuring that computing the partial sum of a block prefix at some level reduces to one instance of the same problem in the next level. 
All details are reported in the following Theorem. 

\begin{theorem}\label{th:PS}
	Let $\gamma$ be the size of an attractor for a string $S$ of length $n$ over an integer alphabet. Then, for all $2\leq \tau \leq n/\gamma$, we can store a data structure of size $O(\tau \gamma\log_\tau(n/\gamma))$ supporting \texttt{partial sum} queries on $S$ in $O(\log_\tau(n/\gamma))$ time. 
\end{theorem}
\begin{proof}
	For simplicity, we assume that $\gamma$ divides $n$ and that $\log_\tau(n/\gamma)$ is an integer. Our structure is composed of $\log_\tau(n/\gamma)+1$ levels. Level $j\geq 0$ contains a set of blocks (i.e. text substrings), each of length $\ell_j = n/(\gamma\cdot \tau^j)$ and defined as follows. In the first level, number $0$, our blocks are the $\gamma$ contiguous and non-overlapping $S$-substrings of length $n/\gamma$: $B_{0,k} = S[(k-1)\cdot(n/\gamma)+1..k\cdot(n/\gamma)]$, for $k=1,\dots, \gamma$. At level $j\geq 1$, blocks are instead centered around attractor elements. Let $i$ be an attractor position. Then, at level $j\geq 1$ we store the $2\tau$ blocks $\overleftarrow{B}_{i,j,k} = S[i-\ell_j \cdot k..i-\ell_j \cdot (k-1)-1]$, for $k=1,\dots, \tau$, and $\overrightarrow{B}_{i,j,k} = S[i+1+\ell_j \cdot (k-1)..i+\ell_j \cdot k]$, for $k=1,\dots, \tau$ (note: $\overleftarrow{B}_{i,j,k}$ are on the left of attractor position $i$, blocks $\overrightarrow{B}_{i,j,k}$ are on its right, and none of the blocks intersects $i$). 
	
	Note: clearly, we do not explicitly store the text substrings associated with each block. Each block will store just a constant amount of information (detailed below) that will be used to answer partial sum queries. However, letting $B$ be a block, to simplify notation in the following we will also use the symbol $B$ to indicate the substring represented by the block $B$. In this sense, $|B|$ will refer to the substring's length. The use will be clear from the context. 
	 
	Each block $B$ at level $j\geq 0$ stores a pointer to one of its occurrences (as a string) at level $j+1$ crossing an attractor position (at least one occurrence of this kind exists by definition of $1$-adjacent attractor). This pointer is simply a pair $(i,v)$, where $i$ is the attractor position and $v$ is such that $S[i-v..i-v+|B|-1] = B$. Crucially, note that the substring $S[i-v..i-v+|B|-1]$ is completely covered by contiguous and non-overlapping blocks of length $|B|/\tau$ at level $j+1$, except possibly position $S[i]$ that is not included in any of those blocks (this will be used to devise a recursive strategy).
	
	In the last level $j=\log_\tau(n/\gamma)$ (where the block size is $\ell_{\log_\tau(n/\gamma)}=1$) we explicitly store in $2\gamma\tau$ words the strings representing the blocks. 
	We also store the character $S[i]$ under each attractor position $i$. Note that $S[i]$ can be retrieved in constant time from $i$ using, e.g. perfect hashing.  
		
	We  associate to each block some partial sum information. 
	To simplify notation, let $sum(B)$ denote the sum of all integers in the string $B$.
	
	\begin{itemize}
		\item[(a)] For each block $B$ at any level, let $(i,v)$ be the pointer associated with it. Then, we associate to $B$ the value $sum(B[1..v])$.
		\item[(b)] At level 0, each block $B_{0,k}$ with $k = 1, \dots, \gamma$, stores $sum(B_{0,1} \cdots B_{0,k-1})$, i.e. the sum of all integers preceding the block in the input string (this value is 0 for $B_{0,1}$).
		\item[(c)] At levels $j\geq 1$, each block $B$ stores $sum(B)$.
		\item[(d)] Blocks $\overleftarrow{B}_{i,j,k}$ moreover store the partial sum $sum(\overleftarrow{B}_{i,j,k} \cdots \overleftarrow{B}_{i,j,1})$, for $k=1,\dots, \tau$, while blocks $\overrightarrow{B}_{i,j,k}$ store the partial sum $sum(\overrightarrow{B}_{i,j,1} \cdots \overrightarrow{B}_{i,j,k})$, for $k=1,\dots, \tau$. 
	\end{itemize}
	
	Overall, we store $O(1)$ words per block. It follows that our structure fits in $O(\tau \gamma\log_\tau(n/\gamma))$ words. We now show how to efficiently answer partial sum queries using this information. 
	
	To answer $S.psum(v)$ we proceed as follows. Let $k = \lfloor (v-1)/(n/\gamma)\rfloor +1$ and $v' = \left((v-1)\ \rm{mod}\ (n/\gamma)\right)+1$. Then, $S.psum(v) = sum(B_{0,1}, \dots, B_{0,k-1}) + B_{0,k}.psum(v')$. The first term $sum(B_{0,1}, \dots, B_{0,k-1})$ is explicitly stored (read point (b) above). To compute $B_{0,k}.psum(v')$, note that this is a block prefix; we now show how to compute the sum of the integers in the prefix of any block at level $j\geq 0$ by reducing the problem to that of computing the sum of the integers in a prefix of a block at level $j+1$. The answer in the last level $j=\log_\tau(n/\gamma)$ (where the block size is $\ell_{\log_\tau(n/\gamma)}=1$) can be obtained in constant time since we explicitly store the integers contained in the blocks.
	
	Let us show how to compute $sum(B[1..t])$ at level $j\geq 0$, for some $t \leq \ell_j$ and some level-$j$ block $B$. First, we map $B[1..t]$ to level $j+1$ using its associated pointer $(i,v)$. We distinguish two cases. 
	
	\begin{itemize}
		\item[(1)] If $t > v$, then $sum(B[1..t]) = sum(B[1..v]) + sum(B[v+1..t])$. The first term, $sum(B[1..v])$ is explicitly stored (read point (a) above). Note that the string appearing in the second term, $B[v+1..t]$, prefixes $S[i]\cdot \overrightarrow{B}_{i,j+1,1} \cdots \overrightarrow{B}_{i,j+1,\tau}$. We can therefore decompose $B[v+1..t]$ into $S[i]$, followed by a (possibly empty) prefix of $d = \lfloor (t-v-1)/\ell_{j+1} \rfloor$ blocks of length $\ell_{j+1}$ --- i.e. the prefix $\overrightarrow{B}_{i,j+1,1}, \dots, \overrightarrow{B}_{i,j+1,d}$ ---, followed by a  (possibly empty) suffix of length $l = (t-v-1)\ \rm{mod}\ \ell_{j+1}$. We can retrieve in constant time the sum (explicitly stored, see point (d) above) of the integers contained in the prefix of full blocks, as well as the value $S[i]$. As far as the remaining suffix of $B[v+1..t]$ is concerned, note that it coincides with the block prefix $\overrightarrow{B}_{i,j+1,d+1}[1..l]$ and we can thus compute the corresponding partial sum by recursing our strategy.
		\item[(2)] If $t \leq v$, then $sum(B[1..t]) = sum(B[1..v]) - sum(B[t+1..v])$. The first term, $sum(B[1..v])$ is explicitly stored (read point (a) above). The second term can be computed with a strategy completely symmetric to that described in point (1). Let $d = \lfloor(v-t)/\ell_{j+1}\rfloor$ and $l=(v-t)\ \rm{mod}\ \ell_{j+1}$. We decompose $B[t+1..v]$ into the prefix $\overleftarrow{B}_{i,j+1,d+1}[(\ell_{j+1}-l+1)..\ell_{j+1}]$ (note: this is a block suffix) followed by the suffix $\overleftarrow{B}_{i,j+1,d}\cdots \overleftarrow{B}_{i,j+1,1}$.
		The sum of integers in the suffix of full blocks is explicitly stored (point (d) above), so we are left with the problem of computing the sum in $\overleftarrow{B}_{i,j+1,d+1}[(\ell_{j+1}-l+1)..\ell_{j+1}]$, which is a block suffix. Since we explicitly store $sum(\overleftarrow{B}_{i,j+1,d+1})$ (point (c) above), we can, also in this case, reduce the problem to that of computing the sum in a prefix of a block at level $j+1$ (i.e. the block prefix $\overleftarrow{B}_{i,j+1,d+1}[1..\ell_{j+1}-l]$) with a simple subtraction.		
	\end{itemize}
	
	The strategy above described allows us to compute $S.psum(v)$ with a single descent from the first to last level. At each level we spend constant time. It follows that the overall procedure terminates in $O(\log_\tau(n/\gamma))$ time.	
\end{proof}

For $\tau=\log^{\epsilon}n$ and any constant $\epsilon>0$ our structure takes 
$O\left (\gamma \log^\epsilon n\log(n/\gamma) /\log\log n\right)$
words of space and answers queries in $O(\log(n/\gamma)/\log\log n)$ time. This matches the lower bound stated in Theorem \ref{cor:low_bound_compressors}.

\subsection{Rank}

Clearly, on binary strings it holds that $S.rank_1(i) = S.psum(i)$ so our problem is already solved in this case by Theorem \ref{th:PS}. Given a string $S$ over a generic alphabet of size $\sigma$, we can solve $S.rank_c()$ as follows. We build $\sigma$ bit-strings $S_c$, one for each alphabet character $c$, defined as $S_c[i] = 1$ if and only if $S[i]=c$. It is easy to verify that, if $S$ has an attractor $\Gamma$, then $\Gamma$ is an attractor also for $S_c$ (repetitions are preserved). Then, we build our structure of Theorem \ref{th:PS} on each $S_c$ using $\Gamma$ as attractor and compute $S.rank_c(i) = S_c.rank_1(i)$ (we can associate each $c$ to its structure on $S_c$ in constant time by using perfect hashing). We obtain: 

\begin{theorem}\label{th:rank}
	Let $\gamma$ be the size of a string attractor for a string $S$ of length $n$ over an alphabet of size $\sigma$. Then, for all $2\leq \tau \leq n/\gamma$, we can store a data structure of size $O(\tau \sigma \gamma\log_\tau(n/\gamma))$ supporting \emph{rank} queries on $S$ in $O(\log_\tau(n/\gamma))$ time. 
\end{theorem}

For $\tau=\log^{\epsilon}n$ and any constant $\epsilon>0$, our structure takes 
$O\left (\sigma\gamma \log^\epsilon n\log(n/\gamma) /\log\log n\right)$
words of space and answers queries in $O(\log(n/\gamma)/\log\log n)$ time. This running time matches the lower bound stated in Theorem \ref{cor:low_bound_compressors} when $\sigma \in O(\rm {polylog}\ n)$.

\subsection{Select}\label{sec:select}

We first consider the binary case and $select_1$ queries.
We use a straightforward reduction from $select_1$ to \emph{partial sum} queries that blows up the attractor size only by a constant factor.

We assume for simplicity that our input bit-vector $S$ ends with bit 1. If this is not the case, removing the trailing zeros from $S$ does not change the answer of any $select_1$ query and does not increase the attractor's size. 
Let $S = 0^{x_1-1}1\dots 0^{x_m-1}1$, with $x_i\geq 1$ for $i=1, \dots, m$, and define $S' = x_1\dots x_m$.

\begin{lemma}\label{lem:attr select}
	If $S$ has an attractor $\Gamma$ of size $\gamma$, then $S'$ has an attractor of size at most $2\gamma+1$.
\end{lemma}
\begin{proof}
	We build an attractor $\Gamma'$ for $S'$ as follows. 
	Note that, to build $S'$, we partition $S$ in blocks: each block is formed by a sequence of zeros terminated by a bit set. Given a position $i\in[1,|S|]$, we say that $i'$ is the \emph{corresponding position} of $i$ in $S'$ iff $i$ belongs to the $i'$-th block. 
	Starting with $\Gamma'=\{1\}$, for every $i\in \Gamma$ we insert in $\Gamma'$ the positions  $i'$ and $i'+1$ (unless they fall outside the range $[1,|S'|]$), where $i'$ is the corresponding position of $i$ in $S'$. More formally, 
	if $S[i]=0$ then we insert $S.rank_1(i)+1$ and $S.rank_1(i)+2$ in $\Gamma'$. Otherwise, if $S[i]=1$ then we insert $S.rank_1(i)$ and $S.rank_1(i)+1$ in $\Gamma'$. Clearly, $|\Gamma'| \leq 2|\Gamma|+1 = 2\gamma+1$.
	
	Consider any substring $S'[i..j]$. We prove that $S'[i..j]$ has an occurrence $S'[i'..j']$ (possibly, $i'=i$ and $j'=j$) such that $S'[i'..j']$ crosses an element of $\Gamma'$. This will prove our claim.
	
	If $i=1$, then $S'[i..j]$ crosses position $1\in\Gamma'$  and we are done.
	Otherwise, we focus on sequence $S'[i-1..j] = y_1 y_2 \dots y_{j-i+2}$.
	Note that this is the sequence of exponents of zeros in a subsequence of $S$ entirely covered by blocks: $S[S.select_1(i-2)+1..S.select_1(j)] = 0^{y_1}1\dots 0^{y_{j-i+2}}1$, where we take $S.select_1(0)=0$.
	By definition of $\Gamma$, this substring of $S$ has an occurrence $S[i'',j'']$ crossing an element of $\Gamma$. However, note that this occurrence is not necessarily preceded by a bit set. As a consequence, we can only state that the corresponding subsequence of $S'$, i.e. $S'[t..S.rank_1(j'')]$ with $t=S.rank_1(i'')+1$ if $S[i'']=0$ and $t=S.rank_1(i'')$ otherwise, is such that $S'[t..S.rank_1(j'')] = w  y_2 \dots y_{j-i+2}$, with $w\geq y_1$. 
	Since $S[i'',j'']$ crosses an element of $\Gamma$ then, by the way we defined $\Gamma'$, either $S'[S.rank_1(j'')]$ or two adjacent characters $S'[k,k+1]$, with $t \leq k < S.rank_1(j'')$, cross an element of $\Gamma'$. Then, this means that $S'[t+1..S.rank_1(j'')] = y_2 \dots y_{j-i+2} = S'[i..j]$ crosses an element of $\Gamma'$. This concludes our proof.
\end{proof}

At this point, the solution for select is immediate: we build the structure of Theorem \ref{th:PS} on the sequence $S' = x_1\dots x_m$ using the attractor of Lemma \ref{lem:attr select}, and simply note that $S.select_1(i) = S'.psum(i)$.

Given a string $S$ over a generic alphabet of size $\sigma$, we can solve $S.select_c()$ as follows. We build $\sigma$ bit-strings $S_c$, one for each alphabet character $c$, defined as $S_c[i] = 1$ if and only if $S[i]=c$. As seen in the previous section, an attractor for $S$ is also an attractor for $S_c$. We build our structure solving $select_1$ on each $S_c$ and compute $S.select_c(i) = S_c.select_1(i)$. We obtain: 

\begin{theorem}\label{th:select}
	Let $\gamma$ be the size of a string attractor for a string $S$ of length $n$ over an alphabet of size $\sigma$. Then, for all $2\leq \tau \leq n/\gamma$, we can store a data structure of size $O(\tau \sigma \gamma\log_\tau(n/\gamma))$ supporting \emph{select} queries on $S$ in $O(\log_\tau(n/\gamma))$ time.
\end{theorem}

For $\tau=\log^{\epsilon}n$ and any constant $\epsilon>0$, our structure takes 
$O\left (\sigma\gamma \log^\epsilon n\log(n/\gamma) /\log\log n\right)$
words of space and answers queries in $O(\log(n/\gamma)/\log\log n)$ time. This running time matches the lower bound stated in Theorem \ref{cor:low_bound_compressors} when $\sigma \in O(\rm {polylog}\ n)$.

\subsection{Predecessor and Tree Navigation}

Let $U\subseteq [1,n]$ be a set of size $m$, and let $S$ be a binary string of length $n$ such that $S[i]=1$ iff $i\in U$ with attractor of size $\gamma$. Using the solutions for \emph{rank} and \emph{select} seen in the previous sections, we can easily support \emph{predecessor} on $U$ in $O(\gamma\ \rm{polylog}\ n)$ space and $O(\log (n/\gamma)/\log\log n)$ time (optimal by Corollary \ref{cor:low_bound_compressors}). In an extended version of this paper we will show that we can improve this upper bound (both in space and time) on sparse sets, achieving $O(\gamma \log^{1+\epsilon} m /\log\log m) = O(\gamma\ \rm{polylog}\ m)$ space and $O(\log m/\log\log m)$ query time. This solution is analogous to that used in Theorem \ref{th:PS} but, in addition, uses fusion trees to accelerate local predecessor queries.

To conclude, one can obtain fast navigational queries on attractor-compressed trees by combining the SLP-based implementation of balanced-parentheses operations (1-7) (see Section \ref{sec:lower tree}) described by Bille et al.~\cite[Lem. 8.2, 8.3]{bille2015random} with the SLP of~\cite[Thm. 3.14]{kempa2018roots}, built on the balanced-parentheses representation of the tree. This immediately  yields $O(\log n)$-time navigation within $O(\gamma\log^2(n/\gamma))$ words of space. To reduce this running time to the optimal $O(\log n/\log\log n)$, we note that one could increase the arity of the SLP to $\log^\epsilon n$ as described in \cite[Thm. 2]{belazzougui2015access}. Space could be further reduced by employing the RLSLP---of size $O(\gamma\log(n/\gamma))$---described in \cite{navarro2018faster} and adapting the algorithms to work on run-length SLPs. We will cover these improvements in an extended version of this paper.

\bibliography{paper}

\begin{thebibliography}{10}

\bibitem{beame2002optimal}
Paul Beame and Faith~E Fich.
\newblock Optimal bounds for the predecessor problem and related problems.
\newblock {\em Journal of Computer and System Sciences}, 65(1):38--72, 2002.

\bibitem{BN14}
D.~Belazzougui and G.~Navarro.
\newblock Optimal lower and upper bounds for representing sequences.
\newblock {\em ACM Transactions on Algorithms}, 11(4):article 31, 2015.

\bibitem{belazzougui2009monotone}
Djamal Belazzougui, Paolo Boldi, Rasmus Pagh, and Sebastiano Vigna.
\newblock Monotone minimal perfect hashing: searching a sorted table with o (1)
  accesses.
\newblock In {\em Proceedings of the twentieth annual ACM-SIAM symposium on
  Discrete algorithms}, pages 785--794. SIAM, 2009.

\bibitem{belazzougui2015access}
Djamal Belazzougui, Patrick~Hagge Cording, Simon~J Puglisi, and Yasuo Tabei.
\newblock Access, rank, and select in grammar-compressed strings.
\newblock In {\em Algorithms-ESA 2015}, pages 142--154. Springer, 2015.

\bibitem{belazzougui2015queries}
Djamal Belazzougui, Travis Gagie, Pawel Gawrychowski, Juha K{\"a}rkk{\"a}inen,
  Alberto Ord{\'o}nez, Simon~J Puglisi, and Yasuo Tabei.
\newblock Queries on lz-bounded encodings.
\newblock In {\em Data Compression Conference (DCC), 2015}, pages 83--92. IEEE,
  2015.

\bibitem{benoit2005representing}
David Benoit, Erik~D Demaine, J~Ian Munro, Rajeev Raman, Venkatesh Raman, and
  S~Srinivasa Rao.
\newblock Representing trees of higher degree.
\newblock {\em Algorithmica}, 43(4):275--292, 2005.

\bibitem{bille2015random}
Philip Bille, Gad~M Landau, Rajeev Raman, Kunihiko Sadakane, Srinivasa~Rao
  Satti, and Oren Weimann.
\newblock Random access to grammar-compressed strings and trees.
\newblock {\em SIAM Journal on Computing}, 44(3):513--539, 2015.

\bibitem{blumer1987complete}
Anselm Blumer, Janet Blumer, David Haussler, Ross McConnell, and Andrzej
  Ehrenfeucht.
\newblock Complete inverted files for efficient text retrieval and analysis.
\newblock {\em Journal of the ACM}, 34(3):578--595, 1987.

\bibitem{burrows1994block}
Michael Burrows and David~J. Wheeler.
\newblock A block sorting lossless data compression algorithm.
\newblock Technical Report 124, Digital Equipment Corporation, 1994.

\bibitem{crochemore1997direct}
Maxime Crochemore and Renaud V{\'e}rin.
\newblock Direct construction of compact directed acyclic word graphs.
\newblock In {\em Combinatorial Pattern Matching (CPM)}, pages 116--129.
  Springer, 1997.

\bibitem{fredman1990blasting}
Michael~L Fredman and Dan~E Willard.
\newblock Blasting through the information theoretic barrier with fusion trees.
\newblock In {\em Proceedings of the twenty-second annual ACM symposium on
  Theory of Computing}, pages 1--7. ACM, 1990.

\bibitem{ganardi2018tree}
Moses Ganardi, Danny Hucke, Markus Lohrey, and Eric Noeth.
\newblock Tree compression using string grammars.
\newblock {\em Algorithmica}, 80(3):885--917, 2018.

\bibitem{jacobson1989space}
Guy Jacobson.
\newblock Space-efficient static trees and graphs.
\newblock In {\em Foundations of Computer Science, 1989., 30th Annual Symposium
  on}, pages 549--554. IEEE, 1989.

\bibitem{kempa2018string}
Dominik Kempa, Alberto Policriti, Nicola Prezza, and Eva Rotenberg.
\newblock {String Attractors: Verification and Optimization}.
\newblock In Yossi Azar, Hannah Bast, and Grzegorz Herman, editors, {\em 26th
  Annual European Symposium on Algorithms (ESA 2018)}, volume 112 of {\em
  Leibniz International Proceedings in Informatics (LIPIcs)}, pages
  52:1--52:13, Dagstuhl, Germany, 2018. Schloss Dagstuhl--Leibniz-Zentrum fuer
  Informatik.

\bibitem{kempa2018roots}
Dominik Kempa and Nicola Prezza.
\newblock {At the Roots of Dictionary Compression: String Attractors}.
\newblock In {\em Annual Symposium on Theory of Computing (STOC)}, pages
  827--840. ACM, 2018.

\bibitem{KidaMSTSA03}
T.~Kida, T.~Matsumoto, Y.~Shibata, M.~Takeda, A.~Shinohara, and S.~Arikawa.
\newblock Collage system: A unifying framework for compressed pattern matching.
\newblock {\em Theor. Comput. Sci.}, 298(1):253--272, 2003.

\bibitem{KY00}
John~C. Kieffer and En-Hui Yang.
\newblock Grammar-based codes: {A} new class of universal lossless source
  codes.
\newblock {\em IEEE Transactions on Information Theory}, 46(3):737--754, 2000.

\bibitem{lempel1976complexity}
A.~Lempel and J.~Ziv.
\newblock On the complexity of finite sequences.
\newblock {\em {IEEE} Trans. Information Theory}, 22(1):75--81, 1976.

\bibitem{navarro2016compact}
Gonzalo Navarro.
\newblock {\em Compact data structures: A practical approach}.
\newblock Cambridge University Press, 2016.

\bibitem{navarro2018faster}
Gonzalo Navarro and Nicola Prezza.
\newblock Faster attractor-based indexes.
\newblock {\em arXiv preprint arXiv:1811.12779}, 2018.

\bibitem{navarro2018universal}
Gonzalo Navarro and Nicola Prezza.
\newblock {Universal Compressed Text Indexing}.
\newblock {\em Theoretical Computer Science}, 2018.

\bibitem{ordonez2017grammar}
Alberto Ord{\'o}{\~n}ez, Gonzalo Navarro, and Nieves~R Brisaboa.
\newblock Grammar compressed sequences with rank/select support.
\newblock {\em Journal of Discrete Algorithms}, 43:54--71, 2017.

\bibitem{storer1982data}
James~A. Storer and Thomas~G. Szymanski.
\newblock Data compression via textual substitution.
\newblock {\em Journal of the ACM}, 29(4):928--951, 1982.

\bibitem{CVY13}
Elad Verbin and Wei Yu.
\newblock Data structure lower bounds on random access to grammar-compressed
  strings.
\newblock In {\em Annual Symposium on Combinatorial Pattern Matching}, pages
  247--258. Springer, 2013.

\end{thebibliography}

\appendix

\end{document}